\documentclass[showpacs,amsmath,12pt,prd,onecolumn,]{revtex4-2}
\usepackage{graphicx}
\usepackage{latexsym}
\usepackage{tikz}
\usepackage{epsfig}
\usepackage{amssymb}
\usetikzlibrary{positioning}
\usepackage{color}
\usepackage{hyperref}
\hypersetup{colorlinks,citecolor=blue}
\usepackage{multirow}
\usepackage{amsthm}
\newtheorem{proposition}{Proposition}

\usepackage[utf8]{inputenc}
\usepackage[T1]{fontenc}
\usetikzlibrary{positioning}

\preprint{}

\begin{document}
\title{Energy–Momentum Surfaces: A Differential Geometric Framework for Dispersion Relations}		
\author{Ginés R. Pérez Teruel}
\email{gines.landau@gmail.com}
\affiliation{Consellería de Educación, Universidades y Empleo. Ministerio de Educación y Formación Profesional, Spain.}
\begin{abstract}
We propose a geometric framework where dispersion relations are viewed as parametric surfaces in energy–momentum space. Within this picture, the presence and type of critical points of the surface emerge as clear geometric signatures of kinematical restrictions. The Newtonian relation corresponds to a developable surface with no critical points, reflecting the absence of invariant limits. Special Relativity generates a saddle point and globally negative curvature, encoding the universal light cone. Modified dispersion relations may introduce additional critical points, signaling new invariant energy scales or thresholds. This unifying approach not only recasts known results in a transparent geometric language but also provides a simple diagnostic tool for exploring departures from Lorentz invariance and their physical implications.

\end{abstract}
\maketitle
\section{Introduction}
The Special Theory of Relativity (SR) is one of the cornerstones of modern science. Since its publication in 1905, the predictions of the theory have been successfully tested in a large number of experiments, although scientists still look for possible departures that could reveal the presence of new physics, particularly in the photon sector, with observations of cosmic rays or ultra energetic gamma ray bursts. A detailed compilation of the plethora of experimental tests performed, and the results obtained so far, can be found in \cite{Wu,Kos0}. It is well known that SR, presented in its original form by Einstein, was devoid of a sophisticated mathematical structure. Indeed, its rich underlying geometrical order was only discovered after the seminal works of Hermann Minkowski in 1908 \cite{Min1,Min2}, where it was proved that the theory had a deeper geometrical significance than Einstein himself thought. In this sense, the crucial discovery that space and time are no longer independent entities, but an indivisible part of a four-dimensional continuum, paved the way for the development of General Relativity (GR) \cite{Renn}.

Possible violations of Lorentz invariance have been widely considered in the context of quantum gravity theories, which often predict departures from SR at the Planck scale \cite{Kos1,Kos2,Kos3,Matt,Bluhm,Amel,Tass,Wei,AmelinoCamelia2008}. These departures are usually encoded in modified dispersion relations (MDR), providing an effective phenomenological framework where new scales or thresholds may appear \cite{Jacobson2003,Myers}. Several approaches have explored their geometric meaning, ranging from Finsler geometry to the proposal of Relative Locality \cite{Girelli2007,AmelinoCamelia2011,Liberati2013}, and including the so-called Gravity’s Rainbow \cite{GR}, in which MDR are associated with energy-dependent deformations of the spacetime metric. 

In this work we adopt a different but complementary viewpoint. We revisit the dispersion relation of SR and its possible modifications by treating them as parametric surfaces in energy–momentum space. This construction allows one to analyze their local and global properties with the standard tools of differential geometry, such as tangent spaces, metric tensors, curvature, and geodesics. The relativistic dispersion relation is then reinterpreted as a Lorentzian hyperbolic paraboloid, while Newtonian and modified relations give rise to surfaces with qualitatively different geometry. Within this framework, a central observation emerges: the existence and type of critical points of the surface act as geometric signatures of fundamental kinematical restrictions. The Newtonian case leads to a developable surface with no critical points, consistently with the absence of invariant limits. Relativity, in contrast, introduces a saddle point at the origin together with strictly negative Gaussian curvature, reflecting the emergence of a universal causal structure. Modified dispersion relations can generate additional critical points, which correspond to new invariant momentum or energy scales where the geometry of the mass shell undergoes a qualitative change. 

Thus, while the geometry of the relativistic mass shell is well known, the novelty of our contribution lies in providing a unified geometric perspective where Newtonian, relativistic, and modified dispersion relations are compared on equal footing, and in identifying critical points as simple yet robust geometric markers of the fundamental restrictions that underlie each kinematical framework.
\section{Energy--momentum parametric surfaces}
In order to treat dispersion relations within a differential--geometric framework, we introduce the notion of an energy--momentum parametric surface.  
Formally, a surface is a smooth map
\begin{equation}
\mathbf{r}: U\subseteq \mathbb{R}^{2}\longrightarrow \mathbb{R}^{3},
\label{eq:surface_general}
\end{equation}
defined by two parameters, the energy $E$ and the momentum $p$, with local coordinates
\begin{equation}
\mathbf{r}(E,p)=\big(x(E,p),\,y(E,p),\,z(E,p)\big).
\label{eq:parametric_general}
\end{equation}
Here $x(E,p)$, $y(E,p)$ and $z(E,p)$ are assumed to be continuous and differentiable functions.  

At any point $(E_{0},p_{0})\in U$, the differential of $\mathbf{r}$ is the linear map
\begin{equation}
d\mathbf{r}_{(E_{0},p_{0})}:\mathbb{R}^{2}\to \mathbb{R}^{3},
\label{eq:differential}
\end{equation}
represented by the Jacobian matrix. Its columns are the tangent vectors
\begin{align}
\mathbf{r}_{E}(E_{0},p_{0})&=\frac{\partial \mathbf r}{\partial E}(E_{0},p_{0})
=\big(D_{E}x,\,D_{E}y,\,D_{E}z\big), 
\label{eq:rE_general}\\
\mathbf{r}_{p}(E_{0},p_{0})&=\frac{\partial \mathbf r}{\partial p}(E_{0},p_{0})
=\big(D_{p}x,\,D_{p}y,\,D_{p}z\big).
\label{eq:rp_general}
\end{align}
A point $P=\mathbf{r}(E_{0},p_{0})$ is called \emph{regular} if $\mathrm{rank}\, d\mathbf r_{(E_{0},p_{0})}=2$, i.e. if $\mathbf r_{E}$ and $\mathbf r_{p}$ are linearly independent. In that case the tangent plane $T_{P}S$ is uniquely defined. Otherwise, $P$ is a singular point.

For our purposes, it is convenient to restrict to surfaces of the form
\begin{equation}
\mathbf{r}(E,p)=\big(E,\,p,\,f(E,p)\big),
\label{eq:surface_subclass}
\end{equation}
where the first two coordinates coincide with the physical variables.  
Then the tangent vectors reduce to
\begin{equation}
\mathbf{r}_{E}=(1,0,f_{E}),\qquad 
\mathbf{r}_{p}=(0,1,f_{p}),
\label{eq:rE_rp_reduced}
\end{equation}
with $f_{E}=\partial f/\partial E$, $f_{p}=\partial f/\partial p$.  
The cross product
\begin{equation}
\mathbf n=\mathbf r_{E}\times \mathbf r_{p}=
\big(-f_{E},\,-f_{p},\,1\big)
\label{eq:normal_vector}
\end{equation}
is nonzero everywhere, ensuring that these energy--momentum surfaces are regular by construction and that $T_{P}S$ is always well defined.
As a first illustration we apply this construction to the Newtonian dispersion relation $E=\displaystyle \frac{p^2}{2m}$. In this case the associated surface turns out to be developable, with vanishing intrinsic curvature and no critical points. This simple example will provide a useful baseline to contrast with the relativistic and modified dispersion relations discussed in the following sections.
\section{The Newtonian energy--momentum surface}
\label{sec:newtonian}

As a first illustration, consider the Newtonian dispersion relation
\begin{equation}
E=\frac{p^{2}}{2m}.
\label{eq:newton_disp}
\end{equation}
In our framework this defines the parametric surface
\begin{equation}
\mathbf r(E,p)=\big(E,\,p,\,f_{\mathrm N}(E,p)\big),
\qquad
f_{\mathrm N}(E,p)=E-\frac{p^{2}}{2m}.
\label{eq:newton_surface}
\end{equation}

\subsection{Tangent and normal vectors}
The tangent vectors are
\begin{equation}
\mathbf r_{E}=(1,0,f_{E})=(1,0,1),
\qquad
\mathbf r_{p}=(0,1,f_{p})=\Big(0,1,-\frac{p}{m}\Big),
\label{eq:newton_tangent}
\end{equation}
and the (non-unit) normal vector follows from the cross product
\begin{equation}
\mathbf n=\mathbf r_{E}\times \mathbf r_{p}
=\Big(-1,\,\frac{p}{m},\,1\Big).
\label{eq:newton_normal}
\end{equation}
We will also use
\begin{equation}
W:=\sqrt{1+f_{E}^{2}+f_{p}^{2}}
=\sqrt{\,2+\frac{p^{2}}{m^{2}}\,}.
\label{eq:newton_W}
\end{equation}

\subsection{First and second fundamental forms}
The induced metric (first fundamental form) has coefficients
\begin{align}
g_{EE}&=\mathbf r_{E}\!\cdot\!\mathbf r_{E}=1+f_{E}^{2}=2, 
\label{eq:newton_gEE}\\
g_{Ep}&=\mathbf r_{E}\!\cdot\!\mathbf r_{p}=f_{E}f_{p}
=-\frac{p}{m},
\label{eq:newton_gEp}\\
g_{pp}&=\mathbf r_{p}\!\cdot\!\mathbf r_{p}=1+f_{p}^{2}
=1+\frac{p^{2}}{m^{2}}.
\label{eq:newton_gpp}
\end{align}
In matrix form, the first fundamental form is
\begin{equation}
\mathrm{I}_{\mathrm N}=
\begin{pmatrix}
g_{EE} & g_{Ep}\\[2pt]
g_{pE} & g_{pp}
\end{pmatrix}
=
\begin{pmatrix}
2 & -\dfrac{p}{m}\\[8pt]
-\dfrac{p}{m} & 1+\dfrac{p^{2}}{m^{2}}
\end{pmatrix},
\qquad
\det g_{\mathrm N}=2+\frac{p^{2}}{m^{2}}.
\label{eq:newton_first_form_matrix}
\end{equation}

Since $f_{EE}=0$, $f_{Ep}=0$, $f_{pp}=-1/m$, the second fundamental form
\begin{equation}
h_{ij}=\mathbf r_{ij}\cdot \hat{\mathbf n},
\qquad
\hat{\mathbf n}=\frac{\mathbf n}{\|\mathbf n\|}=\frac{\mathbf n}{W},
\label{eq:newton_secondform_def}
\end{equation}
has coefficients
\begin{equation}
L=\frac{f_{EE}}{W}=0,\qquad
M=\frac{f_{Ep}}{W}=0,\qquad
N=\frac{f_{pp}}{W}=-\frac{1}{m\,W}.
\label{eq:newton_II}
\end{equation}
Thus, in matrix form,
\begin{equation}
\mathrm{II}_{\mathrm N}=
\frac{1}{\sqrt{\,2+\dfrac{p^{2}}{m^{2}}\,}}
\begin{pmatrix}
0 & 0\\[6pt]
0 & -\dfrac{1}{m}
\end{pmatrix}.
\label{eq:newton_second_form_matrix}
\end{equation}

\subsection{Curvatures}
For a graph $z=f(E,p)$ one has
\begin{equation}
K=\frac{f_{EE}f_{pp}-f_{Ep}^{2}}{\big(1+f_{E}^{2}+f_{p}^{2}\big)^{2}},
\qquad
H=\frac{(1+f_{p}^{2})f_{EE}-2 f_{E}f_{p}f_{Ep}+(1+f_{E}^{2})f_{pp}}
{2\big(1+f_{E}^{2}+f_{p}^{2}\big)^{3/2}}.
\label{eq:KH_general_graph}
\end{equation}
Where $K$, $H$ are the Gaussian curvature, and the mean curvature, respectively.
Substituting the different values for $f_{\mathrm N}$, yields
\begin{equation}
K_{\mathrm N}(E,p)=0,
\qquad
H_{\mathrm N}(E,p)=-\frac{1}{\,m\left(2+\frac{p^{2}}{m^{2}}\right)^{3/2}}<0.
\label{eq:newton_KH}
\end{equation}
Thus the Newtonian surface is \emph{developable} (vanishing Gaussian curvature), curved in one direction only.

\subsection{Critical points and discussion}
Critical points are defined by $f_{E}=0$ and $f_{p}=0$. Here
\begin{equation}
f_{E}=1\ \ (\text{never }0), 
\qquad
f_{p}=-\frac{p}{m},
\label{eq:newton_grad}
\end{equation}
so there are \emph{no} critical points. This neatly mirrors the absence of universal kinematical restrictions in Newtonian theory: the energy--momentum surface is intrinsically flat ($K=0$) and exhibits no special points where both first derivatives vanish. In the next section we contrast this baseline with the relativistic case, where a nontrivial saddle critical point and strictly negative Gaussian curvature emerge.

\section{The hyperboloid and its hyperbolic paraboloid embedding: Special Relativity as a parametric surface}

The relativistic dispersion relation
\begin{equation}
E^{2}-p^{2}=m^{2}
\label{eq:dispersion_SR}
\end{equation}
defines the so-called mass shell in energy–momentum space. 
It is well known that this set forms a hyperboloid of constant negative curvature, a geometry extensively studied in the mathematical relativity literature \cite{ONeill,Beem,OneillHyperboloid,Spivak}. 
In particular, O’Neill’s classic monograph \cite{ONeill} and Beem–Ehrlich–Easley \cite{Beem} present the hyperboloid as a standard model of Lorentzian geometry, while the original article by O’Neill \cite{OneillHyperboloid} analyzes its curvature properties explicitly. 
For completeness, the formalism of the first and second fundamental forms we use here follows the standard approach in differential geometry textbooks such as Spivak \cite{Spivak}.

\vspace{0.2cm}

At this point it is important to distinguish between two related but different geometric objects:

\begin{itemize}
  \item \textbf{The mass shell itself} is the locus \eqref{eq:dispersion_SR} in the \((E,p)\) plane (more generally in \((E,\vec p)\)-space). This is a two-sheeted hyperboloid (one for $E>0$, one for $E<0$), each sheet having constant negative curvature. 
This is the standard and physically relevant picture of relativistic kinematics.  
  \item \textbf{The embedding used in our framework} is obtained by writing the dispersion relation as the graph of a function in \(\mathbb{R}^{3}\),
\begin{equation}
\mathbf{r}(E,p)=(E,\;p,\;f(E,p)), 
\qquad 
f(E,p)=E^{2}-p^{2}-m^{2}.
\label{eq:parametric_SR}
\end{equation}
In this representation the surface is no longer a hyperboloid but a \emph{hyperbolic paraboloid}, i.e.~a saddle surface. 
This embedding has variable Gaussian curvature but retains the essential Lorentzian signature locally. 
It is this hyperbolic paraboloid representation that we will analyze with the machinery of differential geometry, since it allows us to treat both relativistic and modified dispersion relations on equal footing.
\end{itemize}
\subsection{Tangent and normal vectors}
For a general surface of the form $\mathbf r(E,p)=(E,p,f(E,p))$, the tangent vectors are
\begin{equation}
\mathbf r_{E}=\Big(1,\,0,\,f_{E}\Big), 
\qquad
\mathbf r_{p}=\Big(0,\,1,\,f_{p}\Big),
\label{eq:tangent_general}
\end{equation}
where $f_{E}=\partial f/\partial E$ and $f_{p}=\partial f/\partial p$.  
The (non-normalized) normal vector is obtained as
\begin{equation}
\mathbf n = \mathbf r_{E}\times \mathbf r_{p}=
\Big(-f_{E},\,-f_{p},\,1\Big).
\label{eq:normal_general}
\end{equation}

In the hyperbolic paraboloid embedding \eqref{eq:parametric_SR}, we have
\begin{equation}
f_{E}=2E, 
\qquad 
f_{p}=-2p,
\label{eq:derivatives_f}
\end{equation}
so that the tangent vectors take the explicit form
\begin{equation}
\mathbf r_{E}=(1,0,2E), 
\qquad 
\mathbf r_{p}=(0,1,-2p),
\label{eq:tangent_SR}
\end{equation}
and the normal vector becomes
\begin{equation}
\mathbf n=(-2E,\;2p,\;1).
\label{eq:normal_SR}
\end{equation}

This dual description —hyperboloid as the physical mass shell, and hyperbolic paraboloid as its embedded graphical representation— will be used throughout. 
It avoids confusion and makes clear that our framework does not replace the hyperboloid by a paraboloid, but rather uses the latter as a convenient embedding to apply the standard machinery of differential geometry.
\subsection{First and second fundamental forms}

In general, the first fundamental form is
\begin{equation}
g_{ij}=\mathbf r_{i}\cdot \mathbf r_{j}, \qquad 
\mathrm{I}=\begin{pmatrix}
g_{EE} & g_{Ep}\\
g_{pE} & g_{pp}
\end{pmatrix}.
\label{eq:first_form_general}
\end{equation}
For $f(E,p)=E^{2}-p^{2}-m^{2}$, the components are
\begin{align}
g_{EE}&=1+f_{E}^{2}=1+4E^{2}, \label{eq:gEE}\\
g_{Ep}&=f_{E}f_{p}=-4Ep, \label{eq:gEp}\\
g_{pp}&=1+f_{p}^{2}=1+4p^{2}. \label{eq:gpp}
\end{align}
Thus,
\begin{equation}
\mathrm{I}=
\begin{pmatrix}
1+4E^{2} & -4Ep \\
-4Ep & 1+4p^{2}
\end{pmatrix}.
\label{eq:first_form_SR}
\end{equation}

The second fundamental form is defined as
\begin{equation}
h_{ij}=\mathbf r_{ij}\cdot \hat{\mathbf n}, \qquad 
\mathrm{II}=\begin{pmatrix}
h_{EE} & h_{Ep}\\
h_{pE} & h_{pp}
\end{pmatrix},
\label{eq:second_form_general}
\end{equation}
where $\hat{\mathbf n}=\mathbf n/\|\mathbf n\|$.  
The second derivatives of $f$ are
\begin{equation}
f_{EE}=2, \qquad f_{pp}=-2, \qquad f_{Ep}=0,
\label{eq:second_derivatives}
\end{equation}
so that
\begin{equation}
\mathrm{II}=
\frac{1}{\sqrt{1+4E^{2}+4p^{2}}}
\begin{pmatrix}
2 & 0 \\
0 & -2
\end{pmatrix}.
\label{eq:second_form_SR}
\end{equation}

\subsection{Curvature and critical points}

The Gaussian curvature is given by
\begin{equation}
K(E,p)=\frac{f_{EE}f_{pp}-f_{Ep}^{2}}{(1+f_{E}^{2}+f_{p}^{2})^{2}}
=-\frac{4}{\big(1+4E^{2}+4p^{2}\big)^{2}}<0,
\label{eq:gaussian_curvature}
\end{equation}
while the mean curvature reads
\begin{equation}
H(E,p)=\frac{(1+f_{p}^{2})f_{EE}+(1+f_{E}^{2})f_{pp}}
{2(1+f_{E}^{2}+f_{p}^{2})^{3/2}}
=\frac{4(p^{2}-E^{2})}{\big(1+4E^{2}+4p^{2}\big)^{3/2}}.
\label{eq:mean_curvature}
\end{equation}

Thus, the Lorentzian paraboloid is everywhere hyperbolic ($K<0$), while the sign of $H$ changes along the null lines $E=\pm p$, corresponding to the massless dispersion relation. A critical point occurs at $(E,p)=(0,0)$, with Hessian determinant
\begin{equation}
\det H = f_{EE}f_{pp}-f_{Ep}^{2}=-4<0,
\label{eq:critical_point}
\end{equation}
confirming its nature as a saddle point.

\begin{figure}[h!]
\centering
\includegraphics[width=0.6\textwidth]{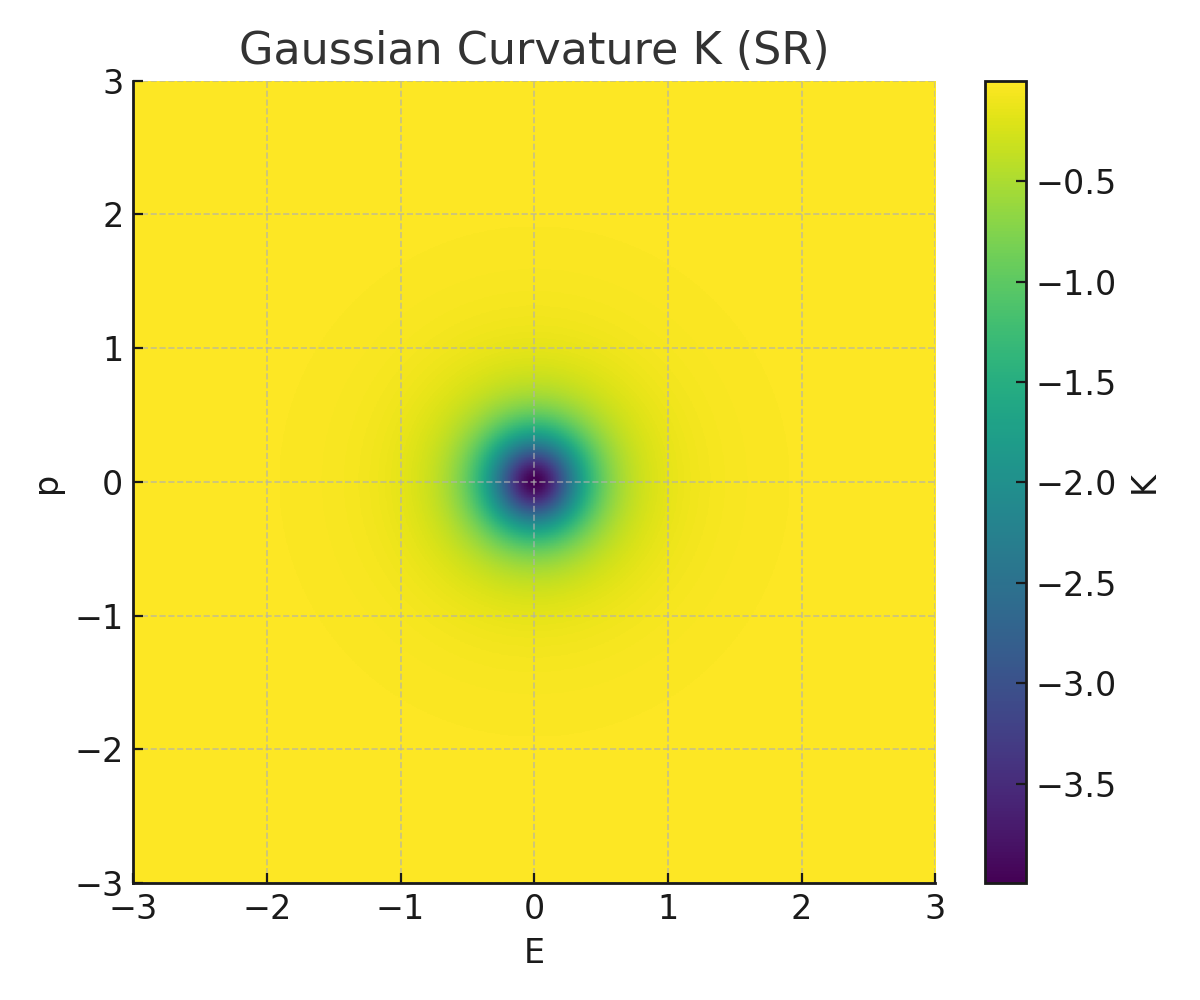}
\caption{Gaussian curvature of the Lorentzian paraboloid. The surface is everywhere negatively curved, with asymptotic directions along $E=\pm p$, corresponding to the massless dispersion relation.}
\label{fig:curvatureSR}
\end{figure}

Although the geometry of the relativistic mass shell is well known, our reinterpretation as a parametric surface allows a unified comparison with the Newtonian case and with modified dispersion relations. In this setting, the existence of a nontrivial critical point and the strictly negative Gaussian curvature appear as compact geometric signatures of the kinematical restrictions specific to Special Relativity.
\section{Relation with Gravity's Rainbow}
\label{sec:GRw}

In Gravity's Rainbow (GRw) \cite{GR}, modified dispersion relations (MDR) are postulated in the factorizable form
\begin{equation}
E^{2}\,f^{2}\!\left(\frac{E}{M_{Pl}}\right) - p^{2}\,g^{2}\!\left(\frac{E}{M_{Pl}}\right) = m^{2},
\label{eq:GRwMDR}
\end{equation}
with smooth, positive functions $f,g$ satisfying $f(0)=g(0)=1$. The functions $f,g$ are often interpreted as deformations of the energy/momentum units, leading to an energy-dependent spacetime metric and, eventually, to modified Einstein equations. 

\subsection{A geometric mapping to the SR mass shell}
\label{subsec:mapping}

Within our framework, \eqref{eq:GRwMDR} admits a direct geometric realization as a parametric surface in $(E,p,z)$ via
\begin{equation}
\mathbf r(E,p) = \big(E,\; p,\; F(E,p)\big), 
\qquad 
F(E,p) := E^{2}f^{2}\!\left(\tfrac{E}{M_{Pl}}\right) - p^{2}g^{2}\!\left(\tfrac{E}{M_{Pl}}\right) - m^{2}.
\end{equation}
Define the smooth change of variables
\begin{equation}
\Phi:\ (E,p)\ \longmapsto\ (u,v),
\qquad
u=E\,f\!\left(\tfrac{E}{M_{Pl}}\right), 
\quad 
v=p\,g\!\left(\tfrac{E}{M_{Pl}}\right).
\label{eq:Phi}
\end{equation}
Assuming $f+E f'\neq 0$ and $g>0$, the Jacobian $J_{\Phi}=(f+E f')\,g$ is nonvanishing and $\Phi$ is a local diffeomorphism. Consider the standard SR surface in $(u,v,z)$,
\begin{equation}
\widehat{\mathbf r}(u,v)=(u,\,v,\,u^{2}-v^{2}-m^{2}).
\end{equation}
Then
\begin{equation}
\mathbf r(E,p)=\widehat{\mathbf r}\!\big(\Phi(E,p)\big),
\end{equation}
and the rainbow mass shell $\{F=0\}$ is mapped onto the SR mass shell $\{u^{2}-v^{2}=m^{2}\}$. In this precise sense, GRw MDR are geometrically equivalent to SR in the $(u,v)$ chart, at the price of a nonlinear reparametrization of energy--momentum.

\subsection{Pullback of the first fundamental form}
\label{subsec:pullback}

Let $\widehat g(u,v)$ denote the first fundamental form of the SR surface $\widehat{\mathbf r}(u,v)$,
\begin{equation}
\widehat g(u,v)=
\begin{pmatrix}
1+4u^{2} & -4uv\\[2pt]
-4uv & 1+4v^{2}
\end{pmatrix}.
\end{equation}
By the chain rule, the metric induced on the rainbow surface in the $(E,p)$ coordinates is the pullback
\begin{equation}
g_{(E,p)}=\Phi^{*}\widehat g
=J_{\Phi}^{\!\top}\,\widehat g(u,v)\,J_{\Phi},
\qquad
J_{\Phi}=
\begin{pmatrix}
\partial_{E}u & \partial_{p}u\\[2pt]
\partial_{E}v & \partial_{p}v
\end{pmatrix}
=
\begin{pmatrix}
f+E f' & 0\\[2pt]
p\,g' & g
\end{pmatrix},
\label{eq:pullback}
\end{equation}
with $u=E f$, $v=p g$, and derivatives $f'=\mathrm d f/\mathrm dE$, $g'=\mathrm d g/\mathrm dE$. In the SR limit $f\equiv g\equiv 1$, one has $J_{\Phi}=\mathbb{I}$ and \eqref{eq:pullback} reduces to the metric previously obtained for the Lorentzian hyperbolic paraboloid.

\subsection{Leading-order expansion in $E/M_{Pl}$}
\label{subsec:LO}

Let $x:=E/M_{Pl}$ and write 
$f(x)=1+\alpha(x)$, 
$g(x)=1+\beta(x)$
with $|\alpha|,|\beta|=O(x)$ near $x=0$. Then
\begin{equation}
u=E\,(1+\alpha),\quad v=p\,(1+\beta),
\qquad
J_{\Phi}=
\begin{pmatrix}
1+\alpha+E\alpha' & 0\\[2pt]
p\beta' & 1+\beta
\end{pmatrix}
+O(x^{2}),
\end{equation}
and
\begin{equation}
\widehat g(u,v)=
\begin{pmatrix}
1+4E^{2}\big(1+2\alpha\big) & -4Ep\big(1+\alpha+\beta\big)\\[2pt]
-4Ep\big(1+\alpha+\beta\big) & 1+4p^{2}\big(1+2\beta\big)
\end{pmatrix}
+O(x^{2}).
\end{equation}
Hence $g_{(E,p)}=\Phi^{*}\widehat g$ produces a transparent, Planck-suppressed deformation of the SR metric on the surface, entirely controlled by $\alpha,\beta$ and their derivatives.

\subsection{Symmetries: linear boosts in $(u,v)$ and nonlinear action in $(E,p)$}
\label{subsec:symmetries}

In $(u,v)$ the Lorentz group acts linearly and preserves $u^{2}-v^{2}=m^{2}$. Transporting this action to $(E,p)$ via $\Phi$ yields a nonlinear representation:
\begin{equation}
(E',p')=\Phi^{-1}\!\Big(\Lambda\cdot \Phi(E,p)\Big),
\qquad
\Lambda\in SO(1,1).
\end{equation}
This is precisely the hallmark of deformed Lorentz symmetry in DSR/GRw scenarios: linear in the rectified chart $(u,v)$, nonlinear back in the physical variables $(E,p)$. Our construction, however, remains purely kinematical in energy--momentum space and does not assume an energy-dependent spacetime metric.

\subsection{Scope and limitations of the rainbow mapping}
\label{subsec:scope}

The diffeomorphic ``straightening'' above relies crucially on the factorizable structure \eqref{eq:GRwMDR} with $f,g=f(E/M_{Pl}),g(E/M_{Pl})$. 
Additive MDR of the form
\begin{equation}
E^{2}=p^{2}+m^{2}+\kappa\,\frac{p^{n}}{M_{Pl}^{\,n-2}},\qquad n>2,
\end{equation}
cannot, in general, be mapped to SR by the simple redefinitions \eqref{eq:Phi}. In such cases, the induced surface departs genuinely from constant curvature, and new geometric features appear (e.g., additional critical points and nonuniform or sign-changing Gaussian curvature). 
Thus, our approach both recovers GRw as a particular case---via an explicit embedding and pullback construction---and provides a criterion to distinguish which MDRs belong to the rainbow class and which lie beyond it.
This classification is, to our knowledge, unique to the present embedding formalism: the traditional mass-shell geometry in Minkowski space does not by itself reveal when a MDR is merely a nonlinear reparametrization of special relativity, nor does it provide a diagnostic to separate trivial (factorizable) from genuinely new (non-factorizable) deformations.
\begin{figure}[t]
\centering
\begin{tikzpicture}[>=stealth,thick,node distance=2cm]
  % Estilo de nodos con sombreado
  \tikzstyle{box}=[draw,rounded corners,align=center,
                   minimum width=3.8cm,minimum height=1cm,
                   text centered]

  % Nodo principal
  \node[box,fill=gray!15] (mdr) {General MDR\\$f(E,p)=0$};

  % Dos ramas (casi debajo del centro)
  \node[box,fill=blue!15,below left=of mdr,xshift=-0.6cm,yshift=-0.5cm] 
       (fact) {Factorizable MDR\\(Rainbow / DSR)};

  \node[box,fill=red!15,below right=of mdr,xshift=0.6cm,yshift=-0.5cm] 
       (nonfact) {Non-factorizable MDR\\(poly, log, exp, trig)};

  % Nodos finales (verticales)
  \node[box,fill=blue!25,below=1.8cm of fact] 
       (sr) {SR hyperboloid\\$u^{2}-v^{2}=m^{2}$};

  \node[box,fill=red!25,below=1.8cm of nonfact] 
       (newgeo) {New geometry:\\variable curvature, critical points, tangency};

  % Flechas
  \draw[->] (mdr.south west) -- (fact.north);
  \draw[->] (mdr.south east) -- (nonfact.north);
  \draw[->] (fact.south) -- node[right,align=center] {diffeo $\Phi$} (sr.north);
  \draw[->] (nonfact.south) -- node[right,align=center] {no diffeo} (newgeo.north);

\end{tikzpicture}
\caption{Classification of modified dispersion relations (MDRs) in the embedding framework. 
Factorizable MDRs (e.g.~Rainbow/DSR) are diffeomorphic to the SR hyperboloid and introduce no new geometry. 
Non-factorizable MDRs (e.g.~polynomial, logarithmic, exponential) cannot be straightened by a diffeomorphism, leading to genuinely new geometric structures.}
\label{fig:classification}
\end{figure}
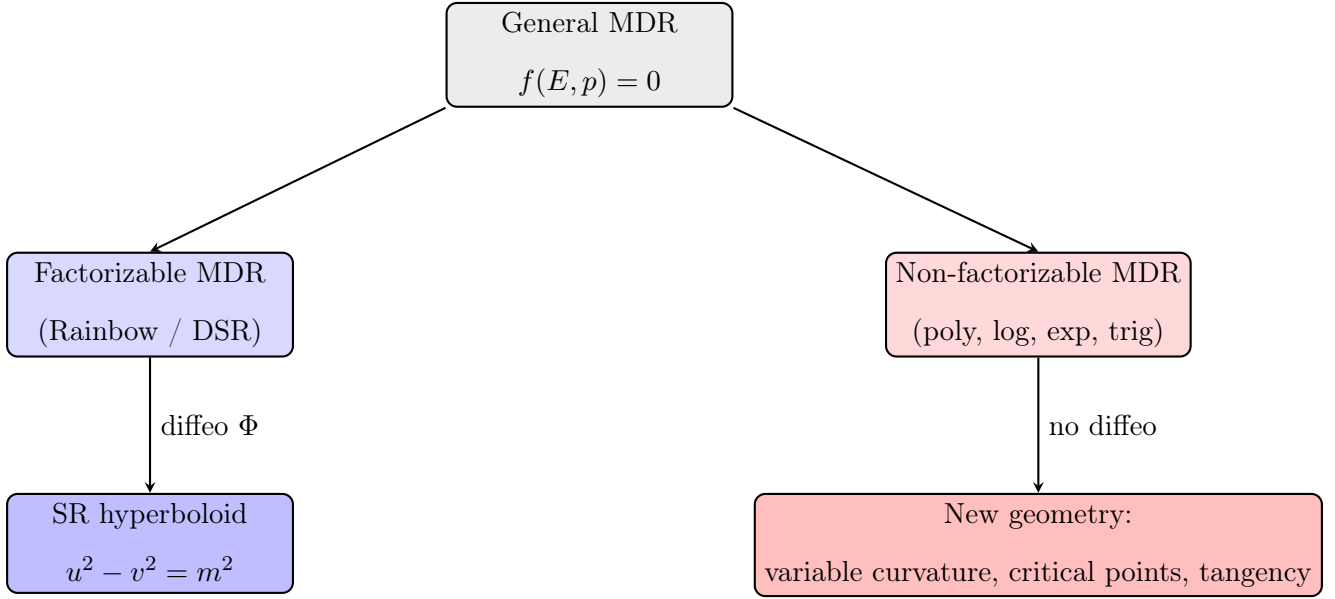

\section{Modified dispersion relations and ultraviolet deformations}
\label{sec:MDR}

In the literature, deviations from the relativistic dispersion relation
\begin{equation}
E^{2}=p^{2}+m^{2}
\label{eq:sr_again}
\end{equation}
at very high energies are not universal. From the perspective of effective field theory, the precise form of such modified dispersion relations (MDRs) depends on the particle species: fermions and bosons can obey different MDRs \cite{Myers}. 

To illustrate our geometric framework, let us consider the generic MDR
\begin{equation}
E^{2}=p^{2}+m^{2}+\kappa_{A}\frac{p^{n}}{M_{Pl}^{\,n-2}},
\label{eq:MDR_relation}
\end{equation}
with $n>2$, $\kappa_{A}$ a dimensionless parameter depending on the particle species $A$, and $M_{Pl}\simeq 1.22\times 10^{19}$ GeV the Planck energy scale \cite{Jacob}. The sign of $\kappa_{A}$ distinguishes ``subluminal'' ($\kappa_{A}<0$) and ``superluminal'' ($\kappa_{A}>0$) scenarios, since the group velocity is $v(E)=\partial E/\partial p$.  

We associate to \eqref{eq:MDR_relation} the parametric surface
\begin{equation}
\mathbf r(E,p)=\big(E,\,p,\,f_{\mathrm M}(E,p)\big),
\qquad
f_{\mathrm M}(E,p)=E^{2}-p^{2}-m^{2}-\kappa_{A}\frac{p^{n}}{M_{Pl}^{\,n-2}}.
\label{eq:MDR_surface}
\end{equation}

\subsection{Tangent and normal vectors}
The tangent vectors are
\begin{equation}
\mathbf r_{E}=(1,0,f_{E})=(1,0,2E),
\qquad
\mathbf r_{p}=(0,1,f_{p})=\Big(0,1,-2p-n\kappa_{A}\frac{p^{n-1}}{M_{Pl}^{\,n-2}}\Big),
\label{eq:MDR_tangent}
\end{equation}
and the (non-unit) normal vector is
\begin{equation}
\mathbf n=\mathbf r_{E}\times\mathbf r_{p}
=\Big(-2E,\,2p+n\kappa_{A}\tfrac{p^{n-1}}{M_{Pl}^{\,n-2}},\,1\Big).
\label{eq:MDR_normal}
\end{equation}
The normalization factor is
\begin{equation}
W=\sqrt{1+f_{E}^{2}+f_{p}^{2}}
=\sqrt{\,1+4E^{2}+\Big(2p+n\kappa_{A}\tfrac{p^{n-1}}{M_{Pl}^{\,n-2}}\Big)^{2}\,}.
\label{eq:MDR_W}
\end{equation}

\subsection{First and second fundamental forms}
The first fundamental form has coefficients
\begin{align}
g_{EE}&=1+f_{E}^{2}=1+4E^{2}, 
\label{eq:MDR_gEE}\\
g_{Ep}&=f_{E}f_{p}=-2E\left(2p+n\kappa_{A}\tfrac{p^{n-1}}{M_{Pl}^{\,n-2}}\right),
\label{eq:MDR_gEp}\\
g_{pp}&=1+f_{p}^{2}=1+\left(2p+n\kappa_{A}\tfrac{p^{n-1}}{M_{Pl}^{\,n-2}}\right)^{2}.
\label{eq:MDR_gpp}
\end{align}
Thus, in matrix form,
\begin{equation}
\mathrm{I}_{\mathrm M}=
\begin{pmatrix}
1+4E^{2} & g_{Ep}\\[4pt]
g_{Ep} & 1+\left(2p+n\kappa_{A}\tfrac{p^{n-1}}{M_{Pl}^{\,n-2}}\right)^{2}
\end{pmatrix}.
\label{eq:MDR_first_form}
\end{equation}

The second derivatives are
\begin{equation}
f_{EE}=2,\qquad
f_{pp}=-2-\frac{n(n-1)\kappa_{A}}{M_{Pl}^{\,n-2}}\,p^{\,n-2},
\qquad
f_{Ep}=0.
\label{eq:MDR_second_derivatives}
\end{equation}
Hence the second fundamental form is
\begin{equation}
\mathrm{II}_{\mathrm M}=\frac{1}{W}
\begin{pmatrix}
2 & 0\\[6pt]
0 & -2-\dfrac{n(n-1)\kappa_{A}}{M_{Pl}^{\,n-2}}\,p^{\,n-2}
\end{pmatrix}.
\label{eq:MDR_second_form}
\end{equation}

\subsection{Curvatures}
From the general formulas for graphs,
\begin{equation}
K=\frac{f_{EE}f_{pp}-f_{Ep}^{2}}{(1+f_{E}^{2}+f_{p}^{2})^{2}},
\qquad
H=\frac{(1+f_{p}^{2})f_{EE}+(1+f_{E}^{2})f_{pp}}
{2(1+f_{E}^{2}+f_{p}^{2})^{3/2}},
\label{eq:MDR_curvature_general}
\end{equation}
we obtain
\begin{equation}
K_{\mathrm M}(E,p)=\frac{2f_{pp}}{(1+4E^{2}+f_{p}^{2})^{2}},
\qquad
H_{\mathrm M}(E,p)=\frac{2(1+f_{p}^{2})+(1+4E^{2})f_{pp}}{2(1+4E^{2}+f_{p}^{2})^{3/2}}.
\label{eq:MDR_curvatures}
\end{equation}
Unlike the relativistic paraboloid, here the Gaussian curvature is not of fixed sign, and can change depending on $p$, $n$, and the sign of $\kappa_{A}$.

\begin{figure}[h!]
\centering
\includegraphics[width=0.6\textwidth]{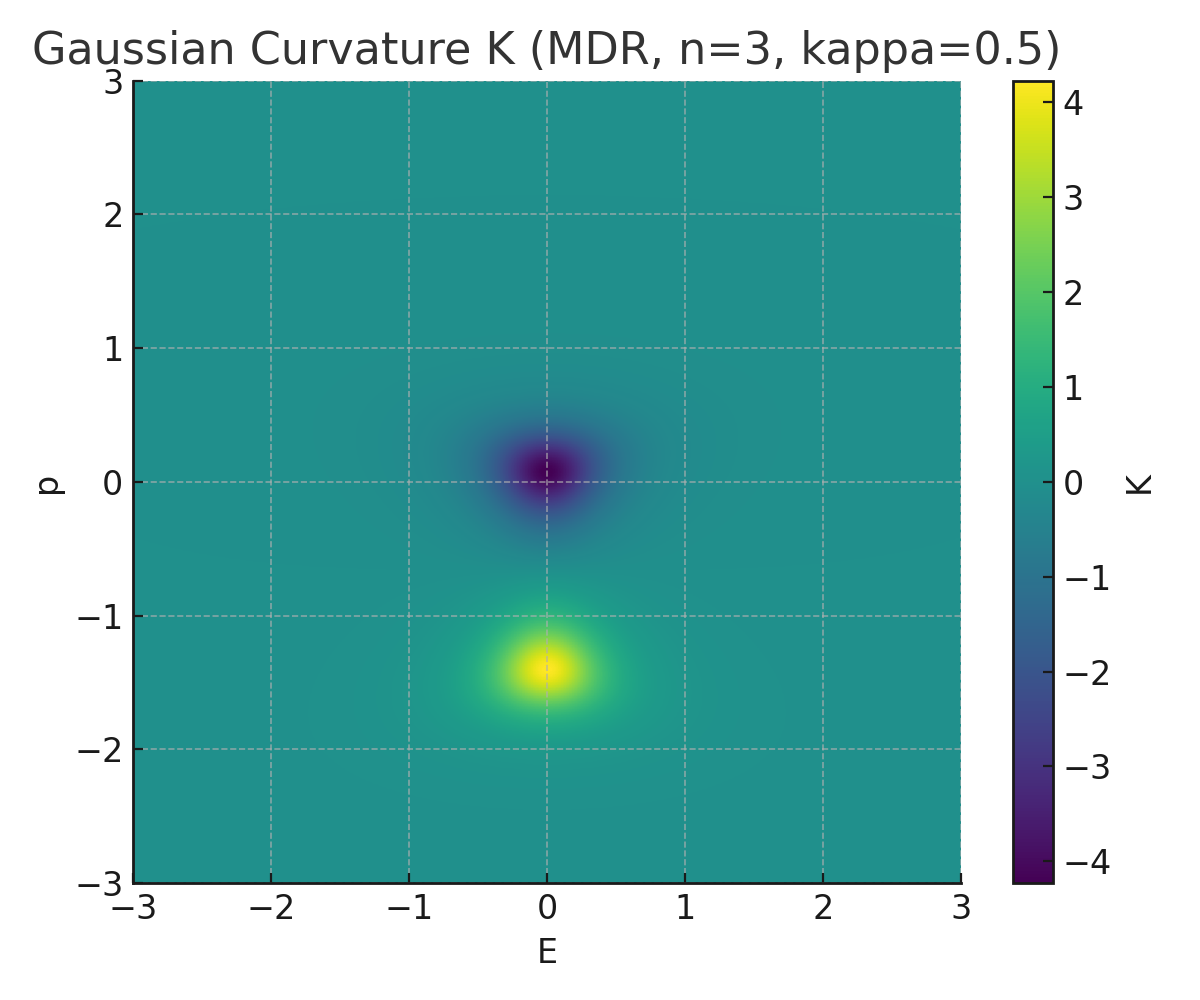}
\caption{Gaussian curvature of the MDR surface for $n=3$, $\kappa_{A}>0$. 
Unlike the Lorentzian paraboloid, the curvature is not everywhere negative: 
it varies with momentum and can change sign, reflecting the richer 
kinematical structure implied by the MDR.}
\label{fig:MDR_curvature}
\end{figure}

\subsection{Critical points and discussion}
Critical points satisfy
\begin{equation}
f_{E}=2E=0,
\qquad
f_{p}=-2p-n\kappa_{A}\frac{p^{n-1}}{M_{Pl}^{\,n-2}}=0.
\label{eq:MDR_critical_eqs}
\end{equation}
The solutions are
\begin{equation}
p_{1}=0,\qquad
p_{2}=\left(\frac{-2M_{Pl}^{\,n-2}}{n\kappa_{A}}\right)^{1/(n-2)}.
\label{eq:MDR_critical_points}
\end{equation}
Thus, unlike the Newtonian and relativistic cases, the MDR surface can exhibit two distinct critical points. For $\kappa_{A}>0$ (superluminal), $p_{2}$ is real only if $n$ is odd, the simplest case being $n=3$, which yields
\begin{equation}
E^{2}=p^{2}+m^{2}+\kappa_{A}\frac{p^{3}}{M_{Pl}},
\label{eq:MDR_cubic}
\end{equation}
a cubic MDR already encountered in effective field theory \cite{Myers}. For $\kappa_{A}<0$ (subluminal), the second critical point exists for all $n>2$.  

In summary, MDR surfaces combine two novel features absent in Newtonian and relativistic dispersion relations: (i) a Gaussian curvature of variable sign, and (ii) additional critical points. Both can be interpreted as geometric signatures of new kinematical restrictions, possibly associated with Planck-scale physics.

\section{Critical points as geometric signatures of kinematical restrictions}
\label{subsec:critical-signatures}

Let $f(E,p)$ define the energy--momentum surface as a graph $z=f(E,p)$ and the mass shell as the level set $f(E,p)=0$.
We call \emph{critical points} the points where $\nabla f=(f_E,f_p)=(0,0)$. In this subsection we argue that:
(i) in Newtonian kinematics no critical points occur, consistently with the developable geometry and the absence of an invariant speed; 
(ii) in SR a saddle-type critical point arises, concomitant with strictly negative Gaussian curvature and the existence of null directions $E=\pm p$;
(iii) in MDR additional critical points can appear, signaling extra kinematical thresholds (new invariant scales) beyond the relativistic one.

\paragraph*{Geometric preliminaries.}
For a surface $z=f(E,p)$, the first and second fundamental forms read
\begin{equation}
\mathrm{I}=
\begin{pmatrix}
1+f_E^2 & f_E f_p\\
f_E f_p & 1+f_p^2
\end{pmatrix}, 
\qquad
\mathrm{II}=
\frac{1}{\sqrt{1+f_E^2+f_p^2}}
\begin{pmatrix}
f_{EE} & f_{Ep}\\
f_{Ep} & f_{pp}
\end{pmatrix}.
\end{equation}
Hence, the Gaussian and mean curvatures are
\begin{equation}
K(E,p)=\frac{f_{EE}f_{pp}-f_{Ep}^2}{\big(1+f_E^2+f_p^2\big)^2}, 
\qquad 
H(E,p)=\frac{(1+f_p^2)f_{EE}-2 f_E f_p f_{Ep}+(1+f_E^2)f_{pp}}{2\big(1+f_E^2+f_p^2\big)^{3/2}}.
\label{eq:KH-general}
\end{equation}
Along the mass shell $f=0$, the (group) velocity follows from implicit differentiation:
\begin{equation}
\frac{dE}{dp}\Big|_{f=0}=-\frac{f_p}{f_E}.
\label{eq:slope}
\end{equation}
If $f_E\neq 0$ (implicit function theorem), the shell can be locally written as $E(p)$; if $f_E=0$ one may instead use $p(E)$ provided $f_p\neq 0$.

\paragraph*{\textbf{Newtonian case}.}
For $f_{\text N}(E,p)=E-\frac{p^2}{2m}$ one has $f_E\equiv 1$, $f_p=-\frac{p}{m}$, hence no critical points exist:
\begin{equation}
f_E=0 \;\text{never holds},\qquad f_p=0 \Rightarrow p=0 \;\text{but}\; f_E\neq 0.
\end{equation}
Moreover $K\equiv 0$ (developable surface). Geometrically and kinematically, the dispersion surface is ``flat'' in the sense of vanishing Gaussian curvature and lacks any special point where both derivatives vanish. This matches the absence of any invariant limiting velocity in Newtonian kinematics.

\paragraph*{\textbf{Special Relativity}.}
For $f_{\text{SR}}(E,p)=E^2-p^2-m^2$ one finds a unique critical point at $(E,p)=(0,0)$:
\begin{equation}
f_E=2E,\quad f_p=-2p,\quad f_{EE}=2,\quad f_{pp}=-2,\quad f_{Ep}=0,
\end{equation}
with Hessian determinant $f_{EE}f_{pp}-f_{Ep}^2=-4<0$, i.e. a saddle.
The curvatures are
\begin{equation}
K_{\text{SR}}(E,p)=-\frac{4}{\big(1+4(E^2+p^2)\big)^2}<0, 
\qquad
H_{\text{SR}}(E,p)=\frac{4(p^2-E^2)}{\big(1+4(E^2+p^2)\big)^{3/2}},
\end{equation}
so the surface is globally hyperbolic (negative $K$). The two asymptotic directions at the saddle are given by the null lines $E=\pm p$ (obtained from the quadratic part $E^2-p^2$ at the origin), which coincide with the massless dispersion $E=\pm p$. 
While the invariant speed $c$ is not encoded by the critical point alone, the existence of a saddle with two asymptotic (null) directions and strictly negative Gaussian curvature provides the geometric substrate for a universal light cone structure: along $f=0$, \eqref{eq:slope} yields $dE/dp=\pm 1$ on the massless branches.

\paragraph*{\textbf{Modified dispersion relations: extra critical points and thresholds}.}
\medskip

Consider, for instance,
\begin{equation}
f_{\text{MDR}}(E,p)=E^2-p^2-m^2-\kappa\,\frac{p^{n}}{M_{Pl}^{\,n-2}},\qquad n>2.
\end{equation}
Critical points satisfy
\begin{equation}
f_E=2E=0,\qquad 
f_p=-2p-\frac{n\kappa}{M_{Pl}^{\,n-2}}\,p^{\,n-1}=0,
\end{equation}
hence $p_1=0$ (the SR saddle) and, provided the RHS admits a nonzero real root,
\begin{equation}
p_2=\left(\frac{-2M_{Pl}^{\,n-2}}{n\kappa}\right)^{\!\frac{1}{n-2}}
\quad (\text{real if }\;\kappa<0\ \text{or}\ n\ \text{odd and }\kappa>0).
\end{equation}
A quadratic Taylor expansion at $(E,p)=(0,p_2)$ gives
\begin{align}
f_{\text{MDR}}(E,p_2+\delta p) &= 
E^2 + \tfrac{1}{2}\,f_{pp}(0,p_2)\,(\delta p)^2 \nonumber \\
&\quad + O\big(E\,\delta p,\,\|(E,\delta p)\|^3\big), \\
f_{pp}(0,p_2) &= -2 
-\frac{n(n-1)\kappa}{M_{Pl}^{\,n-2}}\,p_2^{\,n-2}.
\end{align}
If $f_{pp}(0,p_2)\neq 0$, the Hessian at $(0,p_2)$ has signature $(+,-)$ and the point is again a saddle; if $f_{pp}(0,p_2)=0$ (fine tuning), higher-order terms control the nature of the critical point. 
Along $f=0$, the slope \eqref{eq:slope} reads
\begin{equation}
\frac{dE}{dp}=\frac{2p+\frac{n\kappa}{M_{Pl}^{\,n-2}}\,p^{\,n-1}}{2E}.
\label{eq:slopeMDR}
\end{equation}
At $p=p_2$ one has $f_p=0$ and $E=0$, so \eqref{eq:slopeMDR} is indeterminate ($0/0$): the mass shell develops a \emph{turning point / branching} where the implicit function theorem fails in both parametrizations $E(p)$ and $p(E)$. 
This signals a \emph{kinematical threshold}: a new invariant momentum/energy scale where the dispersion changes regime. 
Geometrically, the additional critical point marks the emergence of a second saddle (or a degenerate criticality), and $K(E,p)$ is deformed with respect to SR:
\begin{equation}
K_{\text{MDR}}(E,p)=\frac{2\big(-2-\frac{n(n-1)\kappa}{M_{Pl}^{\,n-2}}\,p^{\,n-2}\big)}
{\Big(1+4E^2+\big[-2p-\frac{n\kappa}{M_{Pl}^{\,n-2}}\,p^{\,n-1}\big]^2\Big)^2}.
\end{equation}
For $\kappa>0$ (``superluminal'' UV) the numerator becomes more negative and hyperbolicity is enhanced; for $\kappa<0$ (``subluminal'' UV) the numerator can change sign at large $|p|$, producing local elliptic regions ($K>0$). 
Thus, extra critical points and the accompanying distortion of $K$ are \emph{geometric markers} of additional kinematical restrictions beyond SR (new scales/thresholds), in sharp contrast with the Newtonian surface (no critical points, $K=0$ everywhere).
To make the discussion more transparent, Fig.~\ref{fig:curvature-comparison} displays the Gaussian curvature of the three dispersion surfaces. While the Newtonian case is flat ($K=0$ everywhere), special relativity exhibits strictly negative curvature with a single saddle at the origin, and the modified relation develops additional critical points where the curvature changes sign, revealing new kinematical thresholds.
\begin{figure}[t]
\centering
\includegraphics[width=0.9\textwidth]{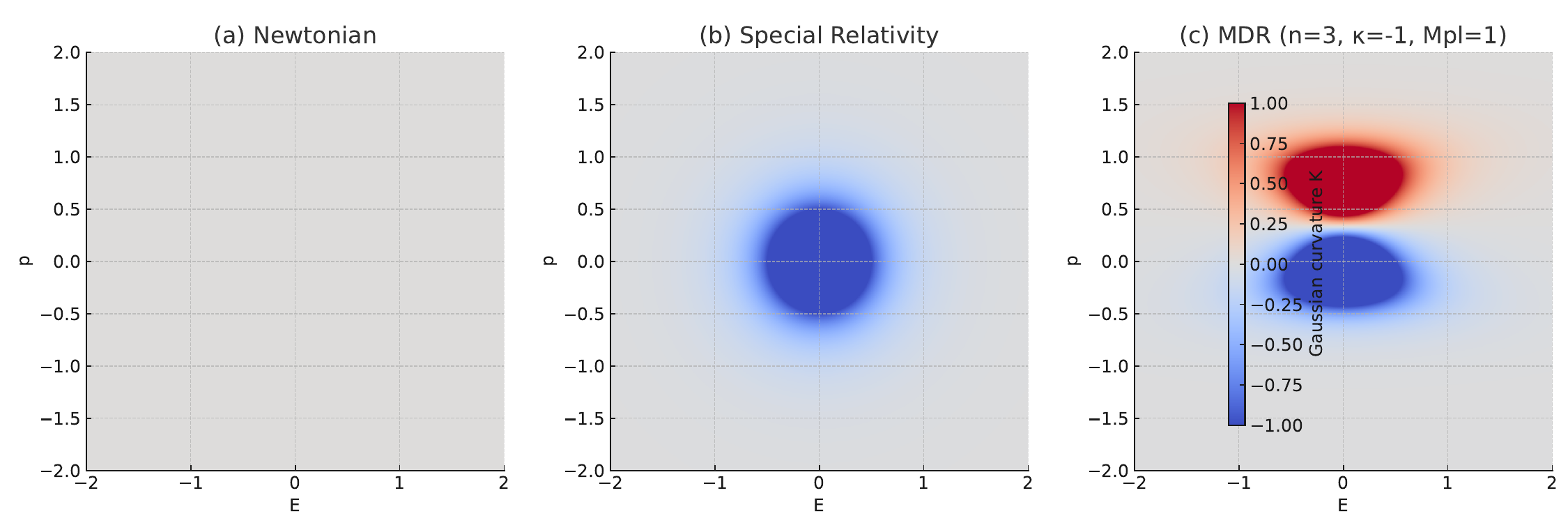}
\caption{Gaussian curvature maps of the dispersion surfaces. 
(a) Newtonian case: flat with $K=0$ and no critical points. 
(b) Special Relativity: strictly negative curvature, with a unique saddle-type critical point at the origin. 
(c) Modified dispersion relation ($n=3$, $\kappa=-1$, $M_{\rm Pl}=1$): curvature of variable sign and an additional critical point, signaling a new kinematical threshold.}
\label{fig:curvature-comparison}
\end{figure}

\paragraph*{Summary.}
Newtonian kinematics yields a developable dispersion surface with no critical points. 
Special Relativity introduces a saddle-type critical point and globally negative Gaussian curvature, consistent with a universal light cone structure (null directions $E=\pm p$ on the mass shell). 
Modified dispersion relations may generate \emph{additional} critical points, where the mass shell develops turning/branching behaviour and the Gaussian curvature pattern is qualitatively altered. 
We therefore propose that the \emph{existence and nature of critical points} in the dispersion surface provide a mathematically precise, geometric signature of fundamental kinematical restrictions: absent in Newtonian theory, present in SR, and possibly enriched in MDR by extra thresholds tied to new invariant scales.
\section{Discussion}

The geometry of the relativistic mass shell, understood as the hyperboloid 
$E^{2}-p^{2}=m^{2}$ in Minkowski space, is of course well known: it has constant 
negative curvature and its geodesics can be characterized by symmetry arguments. 
Our work does not aim to rediscover this classical fact, but rather to reinterpret 
it within a broader geometric framework in which dispersion relations are viewed 
as parametric surfaces in $\mathbb{R}^{3}$. 

This perspective allows for a direct application of the full toolkit of differential 
geometry of surfaces, including the first and second fundamental forms, tangent 
and normal vectors, Gaussian and mean curvatures, and the classification of 
critical points. The novelty of our proposal lies in the comparative and unifying 
power of this method. Indeed:

\begin{itemize}
\item For the Newtonian dispersion relation, the associated surface is developable, 
with vanishing Gaussian curvature and no critical points. 
\item For the relativistic paraboloid, the surface acquires constant negative 
curvature and a single critical point, naturally reflecting the universal 
speed-limit constraint. 
\item For ultraviolet-modified dispersion relations, the surfaces exhibit Gaussian 
curvature of variable sign and typically develop two distinct critical points, 
which can be interpreted as geometric signatures of new kinematical restrictions 
possibly tied to Planck-scale physics.
\end{itemize}

Thus, although the geometry of the relativistic mass shell by itself is not new, 
our contribution is to provide a unified parametric-surface framework that 
encompasses Newtonian, relativistic, and modified dispersion relations on an 
equal footing. Within this framework, the existence and nature of critical 
points, together with the behavior of the Gaussian curvature, emerge as simple 
yet powerful diagnostics of the underlying physical principles. 

We believe this geometric perspective opens a new avenue for studying 
Lorentz-invariance violations and modified dispersion relations, complementing 
other approaches such as Finsler geometry or Rainbow Gravity, and providing 
a compact, visual, and mathematically transparent language for exploring 
possible deformations of special relativity.
\appendix
\section{Geodesics on energy--momentum surfaces}
\label{app:generalMDR}
We consider surfaces given as graphs $z=f(E,p)$, with embedding 
$\mathbf r(E,p)=(E,p,f(E,p))$. The induced metric (first fundamental form) is
\begin{equation}
g_{ij} = \delta_{ij} + f_i f_j,
\qquad
(g_{ij})=
\begin{pmatrix}
1+f_E^2 & f_E f_p\\[2pt]
f_E f_p & 1+f_p^2
\end{pmatrix}, 
\qquad
\Delta:=\det g = 1+f_E^2+f_p^2,
\label{eq:app-metric}
\end{equation}
with inverse
\begin{equation}
(g^{ij})=\frac{1}{\Delta}
\begin{pmatrix}
1+f_p^2 & -\,f_E f_p\\[2pt]
-\,f_E f_p & 1+f_E^2
\end{pmatrix}.
\label{eq:app-inverse}
\end{equation}

For Monge patches there is a convenient compact expression for the Christoffel symbols:
\begin{equation}
\boxed{\quad 
\Gamma^{k}{}_{ij} \;=\; \frac{f_{ij}\, f_{k}}{\Delta}, 
\qquad k,i,j\in\{E,p\},\quad 
\Delta=1+f_E^2+f_p^2\quad}
\label{eq:app-Gamma-compact}
\end{equation}
obtained by inserting $g_{ij}=\delta_{ij}+f_i f_j$ into 
$\Gamma^k{}_{ij}=\tfrac12 g^{k\ell}(\partial_i g_{j\ell}+\partial_j g_{i\ell}-\partial_\ell g_{ij})$.

Let $x^1=E$, $x^2=p$ and $t$ an affine parameter. The geodesic equations are
\begin{equation}
\ddot x^{k}+\Gamma^{k}{}_{ij}\,\dot x^{i}\dot x^{j}=0,
\qquad
\text{i.e.}\quad
\begin{cases}
\ddot E+\dfrac{f_E}{\Delta}\Big(f_{EE}\dot E^{2}+2f_{Ep}\dot E\dot p+f_{pp}\dot p^{2}\Big)=0,\\[10pt]
\ddot p+\dfrac{f_p}{\Delta}\Big(f_{EE}\dot E^{2}+2f_{Ep}\dot E\dot p+f_{pp}\dot p^{2}\Big)=0,
\end{cases}
\label{eq:app-geodesics-general}
\end{equation}
where dots denote derivatives with respect to $t$.

\subsection*{A. Newtonian surface}
\label{app:newton}
Let $f_{\mathrm N}(E,p)=E-\dfrac{p^{2}}{2m}$. Then
\[
f_E=1,\quad f_p=-\frac{p}{m},\quad f_{EE}=0,\quad f_{Ep}=0,\quad f_{pp}=-\frac{1}{m},
\quad \Delta=2+\frac{p^{2}}{m^{2}}.
\]
By \eqref{eq:app-Gamma-compact} the only nonvanishing symbols are
\begin{equation}
\Gamma^{E}{}_{pp}=\frac{f_{pp}f_E}{\Delta}=-\frac{1}{m\,\Delta}, 
\qquad
\Gamma^{p}{}_{pp}=\frac{f_{pp}f_p}{\Delta}=\frac{p}{m^{2}\,\Delta}.
\label{eq:app-newton-Gamma}
\end{equation}
Hence,
\begin{equation}
\ddot E - \frac{1}{m\,\Delta}\,\dot p^{2}=0,
\qquad
\ddot p + \frac{p}{m^{2}\,\Delta}\,\dot p^{2}=0,
\qquad
\Delta=2+\frac{p^{2}}{m^{2}}.
\label{eq:app-newton-geod}
\end{equation}
Since $g_{ij}$ does not depend on $E$, the Killing symmetry $\partial_E$ yields the conserved quantity
\begin{equation}
\Pi_{E} \;=\; g_{EE}\,\dot E + g_{Ep}\,\dot p
\;=\; 2\,\dot E - \frac{p}{m}\,\dot p
\;=\; \text{const},
\label{eq:app-newton-constant}
\end{equation}
which integrates the system by quadratures together with the normalization condition 
$g_{ij}\dot x^{i}\dot x^{j}=\varepsilon$ (e.g. $\varepsilon=1$ for arc-length).

\subsection*{B. Lorentzian paraboloid (Special Relativity)}
\label{app:sr}
Let $f_{\mathrm L}(E,p)=E^{2}-p^{2}-m^{2}$. Then
\[
f_E=2E,\quad f_p=-2p,\quad f_{EE}=2,\quad f_{Ep}=0,\quad f_{pp}=-2,
\quad \Delta=1+4E^{2}+4p^{2}.
\]
From \eqref{eq:app-Gamma-compact}:
\begin{equation}
\Gamma^{E}{}_{EE}=\frac{4E}{\Delta},\quad 
\Gamma^{E}{}_{pp}=-\frac{4E}{\Delta},\quad
\Gamma^{p}{}_{EE}=-\frac{4p}{\Delta},\quad
\Gamma^{p}{}_{pp}=\frac{4p}{\Delta},
\quad \Gamma^{k}{}_{Ep}=0.
\label{eq:app-sr-Gamma}
\end{equation}
Therefore,
\begin{equation}
\ddot E + \frac{4E}{\Delta}\big(\dot E^{2}-\dot p^{2}\big)=0,
\qquad
\ddot p + \frac{4p}{\Delta}\big(\dot p^{2}-\dot E^{2}\big)=0,
\qquad
\Delta=1+4E^{2}+4p^{2}.
\label{eq:app-sr-geod}
\end{equation}
A closed family of solutions is obtained with the linear ansatz \(E(t)=A t+B\) ($A,B$ constants); the first equation gives \(\dot p=\pm A\) and hence
\begin{equation}
p(t)=\pm A t + C,
\qquad
\Rightarrow\quad 
E(p)=\pm p + (B\mp C),
\label{eq:app-sr-lines}
\end{equation}
i.e. straight lines in the $(E,p)$-plane; the case \(B=C\) reproduces the massless branches \(E=\pm p\).

\subsection*{C. Modified dispersion relations (MDR)}
\label{app:mdr}
Let \(f_{\mathrm M}(E,p)=E^{2}-p^{2}-m^{2}-\kappa\,\dfrac{p^{n}}{M_{Pl}^{\,n-2}}\) with \(n>2\).
Then
\[
f_E=2E,\quad 
f_p=-2p-\frac{n\kappa}{M_{Pl}^{\,n-2}}\,p^{\,n-1},\quad
f_{EE}=2,\quad f_{Ep}=0,\quad
f_{pp}=-2-\frac{n(n-1)\kappa}{M_{Pl}^{\,n-2}}\,p^{\,n-2},
\]
and
\[
\Delta=1+f_E^2+f_p^2
=1+4E^{2}+\Big(2p+\frac{n\kappa}{M_{Pl}^{\,n-2}}\,p^{\,n-1}\Big)^{2}.
\]
By \eqref{eq:app-Gamma-compact} the nonvanishing Christoffels are
\begin{equation}
\Gamma^{E}{}_{EE}=\frac{4E}{\Delta},\qquad
\Gamma^{E}{}_{pp}=\frac{2E\,f_{pp}}{\Delta},\qquad
\Gamma^{p}{}_{EE}=\frac{2f_p}{\Delta},\qquad
\Gamma^{p}{}_{pp}=\frac{f_p\,f_{pp}}{\Delta},
\label{eq:app-mdr-Gamma}
\end{equation}
and the geodesic equations become
\begin{equation}
\ddot E + \frac{2E}{\Delta}\Big(2\,\dot E^{2}+f_{pp}\,\dot p^{2}\Big)=0,
\qquad
\ddot p + \frac{f_p}{\Delta}\Big(2\,\dot E^{2}+f_{pp}\,\dot p^{2}\Big)=0.
\label{eq:app-mdr-geod}
\end{equation}

\subsubsection*{C.1 Quadratures near an additional critical point}
Assume the MDR admits a second critical point at $(E,p)=(0,p_2)$, i.e.
$f_E(0,p_2)=0$ and $f_p(0,p_2)=0$ (conditions discussed in the main text).
To obtain a first-order equation for $u(p):=\dfrac{dE}{dp}$, suppose $\dot p\neq 0$
and write $\dot E = u\,\dot p$. Then
\[
\ddot E = u'\,\dot p^2 + u\,\ddot p,
\qquad
Q:=f_{EE}\dot E^2 + 2 f_{Ep}\dot E\dot p + f_{pp}\dot p^2
= \big(f_{EE}u^2 + 2 f_{Ep}u + f_{pp}\big)\,\dot p^2.
\]
Using \eqref{eq:app-geodesics-general} to eliminate $\ddot p$ gives the exact first-order ODE
\begin{equation}
\boxed{\quad
u'(p) + \frac{f_E - u\,f_p}{\Delta}\,\Big(f_{EE}u^2 + 2 f_{Ep}u + f_{pp}\Big)=0,
\qquad u=\frac{dE}{dp}.
\quad}
\label{eq:app-mdr-u-ode}
\end{equation}
This is a Riccati-type equation because $f_E=2E$ depends on $E(p)=\int u\,dp$.
However, \textit{near the additional critical point} $p\simeq p_2$ one has $E\simeq 0$,
so $f_E\simeq 0$ and $\Delta\simeq 1+f_p(p)^2$ depends on $p$ only. With $f_{Ep}=0$
for the MDRs considered, \eqref{eq:app-mdr-u-ode} simplifies to
\begin{equation}
u'(p) - \frac{u\,f_p(p)}{\Delta(p)}\Big(2 u^2 + f_{pp}(p)\Big)=0,
\qquad
\Delta(p)\simeq 1+f_p(p)^2.
\label{eq:app-mdr-u-ode-local}
\end{equation}
Introducing $w(p):=u(p)^{-2}$ yields a \emph{linear} first-order ODE
\begin{equation}
w'(p) + \frac{2 f_p(p) f_{pp}(p)}{\Delta(p)}\,w(p)
= -\,\frac{4 f_p(p)}{\Delta(p)}.
\label{eq:app-mdr-w-linear}
\end{equation}
This admits the closed-form integrating-factor solution
\begin{equation}
w(p) = e^{-\Phi(p)}\!\left[
\,w(p_*) - 4\!\int_{p_*}^{p}\!\frac{f_p(\xi)}{\Delta(\xi)}\,e^{\Phi(\xi)}\,d\xi
\right],
\qquad
\Phi(p)=\int_{p_*}^{p}\frac{2 f_p(\xi) f_{pp}(\xi)}{\Delta(\xi)}\,d\xi,
\label{eq:app-mdr-w-solution}
\end{equation}
where $p_*$ is a reference value near $p_2$. Then $u(p)=w(p)^{-1/2}$ and 
\begin{equation}
E(p)=E(p_*)+\int_{p_*}^{p} u(\xi)\,d\xi
\label{eq:app-mdr-E-by-quadrature}
\end{equation}
provides the geodesic by quadratures in a neighborhood of the additional critical point.
For the explicit MDR family 
$f_p(p)=-2p-\dfrac{n\kappa}{M_{Pl}^{\,n-2}}\,p^{\,n-1}$ and 
$f_{pp}(p)=-2-\dfrac{n(n-1)\kappa}{M_{Pl}^{\,n-2}}\,p^{\,n-2}$, so
both $\Phi$ and the integral in \eqref{eq:app-mdr-w-solution} are explicit functions of $p$.

\medskip
\noindent\textbf{Remarks.}
(i) In the Newtonian case the metric tensor does not depend on $E$, and the conserved quantity $\Pi_E$ in \eqref{eq:app-newton-constant} follows from the associated Killing symmetry.  
(ii) In SR and MDR, $\Delta$ depends on both $E$ and $p$, and no analogous conserved momentum exists \textit{a priori}; nevertheless, the family \eqref{eq:app-sr-lines} provides explicit solutions in SR.  
(iii) All expressions above are independent of the orientation of the normal vector; the sign of the mean curvature may change with orientation, but neither the Gaussian curvature $K$ nor the geodesic equations are affected.

\section{General MDR functions: critical points and curvature}
\label{app:generalMDR}

In the main text we focused on polynomial modified dispersion relations of the form
$E^2=p^2+m^2+\kappa p^n/M_{Pl}^{\,n-2}$. 
Here we summarize the general case, where the correction is an arbitrary smooth function 
$g(p)$ entering as
\[
f(E,p) = E^2 - p^2 - m^2 - g(p).
\]
The following proposition characterizes the critical points and curvature of the corresponding energy–momentum surface $\mathbf{r}(E,p)=(E,p,f(E,p))$ in full generality.
\begin{proposition}[Critical points and curvature of general MDR surfaces]
\label{prop:MDRcrit}
Let a modified dispersion relation be written as
\begin{equation}
f(E,p)=E^2 - p^2 - m^2 - g(p),
\end{equation}
with $g\in C^2(\mathbb{R})$. The associated energy--momentum surface
$\mathbf{r}(E,p)=(E,p,f(E,p))$ satisfies:

\begin{enumerate}
    \item Critical points are the solutions of
    \begin{equation}
    f_E=2E=0,\qquad f_p=-2p-g'(p)=0.
    \end{equation}
    Hence $(E,p)=(0,0)$ is always a critical point (the SR saddle), and any additional critical points correspond to real roots of
    \begin{equation}
    2p+g'(p)=0.
    \end{equation}

    \item The Hessian of $f$ at $(E,p)=(0,p_\star)$ is diagonal,
    \begin{equation}
    Hf=\begin{pmatrix} f_{EE} & f_{Ep}\\ f_{pE} & f_{pp}\end{pmatrix}
    =\begin{pmatrix} 2 & 0\\ 0 & -2-g''(p_\star)\end{pmatrix},
    \end{equation}
    with determinant $\det H=-4-2g''(p_\star)$. Therefore:
    \begin{itemize}
        \item if $\det H<0$ (i.e. $g''(p_\star)>-2$), the point is a saddle;
        \item if $\det H>0$ (i.e. $g''(p_\star)<-2$), the point is a local minimum;
        \item if $\det H=0$, the point is degenerate (higher-order analysis required).
    \end{itemize}

    \item The Gaussian curvature of the graph $z=f(E,p)$ reads
    \begin{equation}
    K(E,p)=\frac{2(-2-g''(p))}{\big(1+4E^2+(2p+g'(p))^2\big)^2}.
    \end{equation}
    Its sign is controlled by $g''(p)$:
    \begin{equation}
    K<0 \iff g''(p)>-2,\qquad
    K>0 \iff g''(p)<-2,\qquad
    K=0 \iff g''(p)=-2.
    \end{equation}
    Thus the loci $g''(p)=-2$ mark curvature sign changes (geometric ``phase transitions'').
\end{enumerate}
\end{proposition}

\begin{proof}[Sketch of proof]
Direct computation gives $f_E=2E$, $f_p=-2p-g'(p)$, so critical points satisfy the above conditions.
The Hessian entries follow from $f_{EE}=2$, $f_{pp}=-2-g''(p)$, $f_{Ep}=0$,
hence $\det H=-4-2g''(p)$.
Finally, inserting $f_E,f_p,f_{EE},f_{pp}$ into the general formula for the Gaussian curvature of a graph $z=f(E,p)$,
\[
K=\frac{f_{EE}f_{pp}-f_{Ep}^2}{\big(1+f_E^2+f_p^2\big)^2},
\]
yields the stated expression. \qedhere
\end{proof}

\begin{acknowledgments}
The author thanks the anonymous referee for constructive criticism and comments that improved the quality of the manuscript.
\end{acknowledgments}


\begin{thebibliography}{99}
\bibitem{Wu} Wei, JJ., Wu, XF. (2022). Tests of Lorentz Invariance. In: Bambi, C., Santangelo, A. (eds) Handbook of X-ray and Gamma-ray Astrophysics. Springer, Singapore. 
\bibitem{Kos0}V.A. Kostelecky, N. Russell, Reviews of Modern Physics 83, 11 (2011). DOI 10.1103/RevModPhys.83.11
\bibitem{Min1} Minkowski 1907–1908, pp. 53–111 The Fundamental Equations for Electromagnetic Processes in Moving Bodies
\bibitem{Min2} Minkowski 1908–1909, pp. "Space and Time"
\bibitem{Renn}J. Renn, M. Schemmel, The Genesis of General Relativity (Springer, Berlin, 2007)
\bibitem{Kos1} V.A. Kostelecky, S. Samuel, Phys. Rev. D39, 683 (1989). DOI 10.1103/
PhysRevD.39.683
\bibitem{Kos2} V.A. Kostelecky, R. Potting, Nuclear Physics B 359, 545 (1991)
\bibitem{Kos3} V.A. Kostelecky, R. Potting, Phys. Rev. D51, 3923 (1995). DOI 10.1103/
PhysRevD.51.3923
\bibitem{Matt} D. Mattingly, Living Reviews in Relativity 8, 5 (2005). DOI 10.12942/
lrr-2005-5
\bibitem{Bluhm} R. Bluhm, Lecture Notes in Physics 702, 191 (2006). DOI 10.1007/
3-540-34523-X\ 8
\bibitem{Amel} G. Amelino-Camelia, Living Reviews in Relativity 16, 5 (2013). DOI 10.
12942/lrr-2013-5
\bibitem{Tass} J.D. Tasson, Reports on Progress in Physics 77(6), 062901 (2014). DOI
10.1088/0034-4885/77/6/062901
\bibitem{Wei} J.J. Wei, X.F. Wu, Frontiers of Physics 16(4), 44300 (2021). DOI 10.1007/
s11467-021-1049-x
\bibitem{Jacob}Jacobson, T., Liberati, S. Mattingly, D. A strong astrophysical constraint on the violation of special relativity by quantum gravity. Nature 424, 1019–1021 (2003). https://doi.org/10.1038/nature01882
\bibitem{MassShellBook1} V. Moretti, Spectral Theory and Quantum Mechanics, Springer (2017).

\bibitem{MassShellBook2} R. Penrose, W. Rindler, Spinors and Space-Time, Vol. 1, Cambridge University Press (1984).

\bibitem{Pfeifer} C. Pfeifer, M. Wohlfarth, Causal structure and electrodynamics on Finsler spacetimes, Phys. Rev. D 92, 084053 (2015).

\bibitem{AmelinoCamelia2008}
G. Amelino-Camelia, ``Quantum Gravity Phenomenology: Status and Prospects,''
arXiv:0806.0339 [gr-qc] (2008).

\bibitem{Jacobson2003}
T. Jacobson, S. Liberati, D. Mattingly, 
``A strong astrophysical constraint on the violation of special relativity by quantum gravity,''
Nature 424, 1019–1021 (2003).
DOI: 10.1038/nature01882

\bibitem{Myers}
R.C. Myers, M. Pospelov,
``Ultraviolet Modifications of Dispersion Relations in Effective Field Theory,''
Phys. Rev. Lett. 90, 211601 (2003).
DOI: 10.1103/PhysRevLett.90.211601

\bibitem{Girelli2007}
F. Girelli, S. Liberati, L. Sindoni, 
``Planck-scale modified dispersion relations and Finsler geometry,''
Phys. Rev. D 75, 064015 (2007).
DOI: 10.1103/PhysRevD.75.064015

\bibitem{AmelinoCamelia2011}
G. Amelino-Camelia, L. Freidel, J. Kowalski-Glikman, L. Smolin,
``The principle of relative locality,''
Phys. Rev. D 84, 084010 (2011).
DOI: 10.1103/PhysRevD.84.084010

\bibitem{Liberati2013}
L. Barcaroli, L.K. Brunkhorst, G. Gubitosi, N. Loret, C. Pfeifer,
``Hamilton geometry: Phase space geometry from modified dispersion relations,''
Int. J. Mod. Phys. D 25, 1650028 (2016).
DOI: 10.1142/S0218271816500288

\bibitem{GR} J. Magueijo, L. Smolin, Gravity’s Rainbow, Class. Quantum Grav. 21 (2004) 1725–1736. arXiv:gr-qc/0305055.

\bibitem{Beem}
J. Beem, P. Ehrlich, K. Easley, \textit{Global Lorentzian Geometry}, 2nd ed., Marcel Dekker (1996).

\bibitem{OneillHyperboloid}
B. O'Neill, ``The geometry of the hyperboloid of one sheet,'' Michigan Math. J. 9, 137–151 (1962).

\bibitem{Spivak}
M. Spivak, \textit{A Comprehensive Introduction to Differential Geometry}, Vol. II, Publish or Perish (1979).

\bibitem{ONeillRiem}
B. O'Neill, \textit{Elementary Differential Geometry}, Revised 2nd ed., Academic Press (2006).

\bibitem{ONeill}
B. O'Neill, \textit{Semi-Riemannian Geometry with Applications to Relativity}, Academic Press (1983).

\end{thebibliography}
\end{document}